\documentclass[letterpaper, 10 pt, conference]{ieeeconf}
\IEEEoverridecommandlockouts
\usepackage{amssymb, amsfonts, mathtools, xargs, tensor, units, cite, stmaryrd, mathrsfs, algpseudocode, algorithm, graphicx, xcolor}
\usepackage[unicode=true, colorlinks=false, hidelinks]{hyperref}
\usepackage{tikz}
\usetikzlibrary{3d, calc}
\usepackage{pgfplots}
\pgfplotsset{compat=newest}
\usetikzlibrary{spy,backgrounds}
\usepackage{multirow}
\usepackage{booktabs}
\usepackage{dblfloatfix}

\usepackage{soul}

\newtheorem{assumption}{Assumption}

\newtheorem{remark}{Remark}

\newtheorem{theorem}{Theorem}

\definecolor{color1}{RGB}{0,128,0}    
\definecolor{color2}{RGB}{255,0,0}    
\definecolor{color3}{RGB}{139,69,19}  
\definecolor{color4}{RGB}{0,0,255}    
\definecolor{color5}{RGB}{152,78,163} 

\DeclareMathOperator\eye{I}

\newcommand\scalemath[2]{\scalebox{#1}{\mbox{\ensuremath{\displaystyle #2}}}}

\allowdisplaybreaks
\usepackage{geometry}
\title{\LARGE \bf An adaptive optimal control approach to monocular depth observability maximization}
\author{Tochukwu Elijah Ogri$^{1}$ \and Muzaffar Qureshi$^{1}$  \and Zachary I. Bell$^{2}$  \and Kristy Waters$^{3}$ \and Rushikesh Kamalapurkar$^{1}$
\thanks{This research was supported by the Air Force Research Laboratories under contract number AFRL-FA8651-23-1-0006. Any opinions, findings, or recommendations in this article are those of the author(s), and do not necessarily reflect the views of the sponsoring agencies.}%
\thanks{$^{1}$ School of Mechanical and Aerospace Engineering, Oklahoma State University, email: {\tt\footnotesize \{tochukwu.ogri, muzaffar.qureshi, rushikesh.kamalapurkar\} @okstate.edu}.}%
\thanks{$^{2}$ Air Force Research Laboratories, Florida, USA, email: {
\tt \footnotesize zachary.bell.10@us.af.mil.}}
\thanks{$^{2}$ Autonomous Vehicles Laboratory, University of Florida, Gainesville, Florida, USA, email: {
\tt \footnotesize watersk@ufl.edu.}}
}
\begin{document}
\maketitle
\thispagestyle{empty}
\pagestyle{empty}

\begin{abstract} 
This paper presents an integral concurrent learning (ICL)-based observer for a monocular
camera to accurately estimate the Euclidean distance to features on a
stationary object, under the restriction that state information
 is unavailable. Using distance estimates, an infinite horizon optimal regulation problem is
solved, which aims to regulate the camera to a goal location while maximizing feature observability. Lyapunov-based stability analysis is used to guarantee exponential convergence of depth estimates and input-to-state stability of the goal location relative to the camera. The effectiveness of the proposed approach is verified in simulation, and a table illustrating improved observability through better conditioning of the regressor is provided.
\end{abstract}

\section{Introduction}\label{section:introduction}

The use of drones and other micro air vehicle systems has seen rapid growth in recent years due to their ability to perform dangerous or complex tasks such as surveillance, search and rescue, and weather monitoring, that are challenging or even impossible for human pilots \cite{SCC.Yu.Sharma.ea2013, SCC.Delaune.Bayard.ea2021}.  
In the absence of state-feedback information
from a positioning system, these robotic systems are forced to navigate, relying solely on local sensing data (e.g., camera images, inertial measurement units, and wheel encoders). The poses of objects in the surrounding environment relative to a robot must thus be determined from sensor data to inform estimation of its pose; otherwise, the performance of the controller may be affected, and the robotic system may fail to achieve its objective.

Accurately estimating the pose of a robot using cameras to reconstruct the environment using scaled Euclidean coordinates of an object is a key challenge, commonly referred to as simultaneous localization and mapping (SLAM)\cite{SCC.Dubbelman.Browning.2015, SCC.Mur.Montiel.ea2015, SCC.Taketomi.Uchiyama.ea2017,SCC.Karrer.Shmuck.ea2018}. A significant challenge in SLAM is determining the scale of
objects in a 2D image, given the loss of depth information. Several image-based methods estimate depth by reconstructing the structure of an object by using multiple images and scale information \cite{SCC.Hartley.Zisserman2003,SCC.Ma.Soatto.ea2004}, or by estimating motion using the camera's linear or angular velocities \cite{SCC.Matthies.Kanade.ea1989,SCC.Jankovic.Ghosh1995,SCC.Soatto.Frezza.ea1996, SCC.Kano.Ghosh.ea2001,SCC.Chiuso.Favaro.ea2002,SCC.Dixon.Fang.ea2003,SCC.Chen.Kano2004, SCC.Karagiannis.Astolfi2005,SCC.Braganza.Dawson.ea2007,SCC.DeLuca.Oriolo.ea2008, SCC.Hu.Aiken.ea2008, SCC.Morbidi.Prattichizzo.ea2009, SCC.Zarrouati.Aldea.ea2012, SCC.Dani.Fischer.ea2012,SCC.Dani.Fischer.ea2011, SCC.Bell.Chen.ea2017} where scales can be recovered using multiple calibrated cameras \cite{SCC.Hartley.Zisserman2003,SCC.Ma.Soatto.ea2004}. However, the performance of motion-based methods is limited when the objects lack parallax between successive camera images. Alternative approaches include the use of extended Kalman filters (EKFs)  \cite{SCC.Matthies.Kanade.ea1989,SCC.Soatto.Frezza.ea1996 , SCC.Kano.Ghosh.ea2001,SCC.Chiuso.Favaro.ea2002,SCC.Reif.Unbehauen1999} for depth estimation. For discrete time systems \cite{SCC.Boutayeb.Rafaralahy.ea1997} developed EKFs with convergence guarantees. However, due to the nature of EKFs their convergence result is local, and the corresponding propagation equations are only valid if the estimates are within a small neighborhood of the actual state. With techniques, such as those proposed in \cite{SCC.Dixon.Fang.ea2003,SCC.Karagiannis.Astolfi2005,SCC.Braganza.Dawson.ea2007,SCC.Hu.Aiken.ea2008, SCC.Zarrouati.Aldea.ea2012, SCC.Dani.Fischer.ea2011}, asymptotic convergence of structure estimation errors is guaranteed, and some of them guarantee exponential convergence of scale estimates \cite{SCC.Jankovic.Ghosh1995,SCC.Chen.Kano2004,SCC.DeLuca.Oriolo.ea2008, SCC.Morbidi.Prattichizzo.ea2009, SCC.Dani.Fischer.ea2012}. However, these methods rely on stringent conditions such as persistence of excitation (PE) or extended output Jacobian (EOJ), which may be difficult to satisfy in practice.

This paper extends the results of \cite{SCC.Bell.Chen.ea2017,SCC.Bell.Deptula.ea2020}, which developed exponentially converging observers using concurrent learning (CL) and integral concurrent learning (ICL) \cite{SCC.Bell.Chen.ea2017,SCC.Chowdhary.Johnson2011,SCC.Chowdhary.Muhlegg2013,SCC.Chowdhary.Yucelen.ea2013} to estimate the Euclidean distance to features on a stationary object in the camera's field of view (FOV) under the assumption that the velocities of the camera are known. The CL-based techniques in \cite{SCC.Bell.Chen.ea2017,SCC.Bell.Deptula.ea2020} guarantee exponential convergence of depth estimates while relaxing the PE assumption in favor of a finite excitation (FE) condition, which can be monitored and verified online. Without sufficient excitation, depth estimation is affected as monocular cameras cannot observe the scale of
objects. 

Excitation conditions require the motion of the camera to be non-parallel to the line joining the camera
and the object \cite{SCC.Bell.Harris.ea2019}.  To achieve such motion, this paper develops an adaptive optimal control scheme that plans velocities for depth observability maximization by penalizing non-orthogonal motion of a monocular camera as it attempts to reach a goal location. The paper demonstrates that feature observability can be improved through velocity planning, without the need for added excitation, by introducing a novel cost function that yields controllers with theoretical stability guarantees. To the best of our knowledge, this is the first study in the current path planning literature where adaptive optimal control
is employed to plan velocities for maximizing depth observability.


\section{Camera Motion Model}\label{section:problemFormulation}	
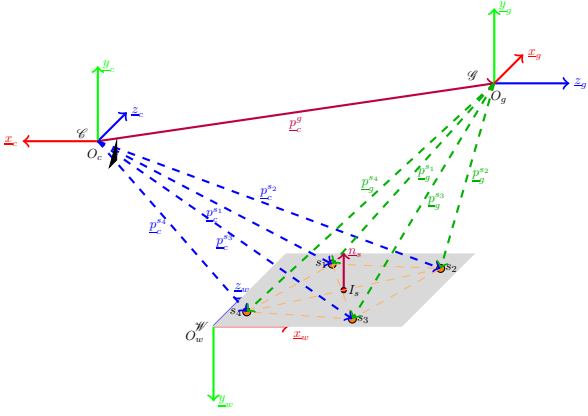
\begin{figure}
    \centering
     \begin{tikzpicture}[scale=0.5, transform shape]

    \coordinate (world) at (0,0,5);
    \node[left] at (world) {$\mathscr{W}$};
    \node[below, text=black] at ($(world)+(-0.5,0,0)$){$O_{w}$}; 
    \draw[->, thick, red] (world) -- ++(2,0,0) node[below right] {$\underline{x}_{w}$};
    \draw[->, thick, blue] (world) -- ++(0,0,-2) node[above] {$\underline{z}_{w}$};
    \draw[->, thick, green] (world) -- ++(0,-2,0) node[right] {$\underline{y}_{w}$};
    
    \coordinate (camera) at (-5,3,0); 
    \node[right] at ($(camera)+(-0.3,0.6,1)$){$\mathscr{C}$};
    \node[below, text=black] at ($(camera)+(0,0,0.2)$){$O_{c}$}; 
    \draw[->, thick, blue] (camera) -- ++(0,0,-2) node[right]{$\underline{z}_{c}$}; 
    \draw[->, thick, green] (camera) -- ++(0,2,0) node[right]{$\underline{y}_{c}$};
    \draw[->, thick, red, rotate around x=90, rotate around y=-90] (camera) -- ++(0,0,2) node[left]{$\underline{x}_{c}$}; 
    
    \filldraw[gray!30] (0,0,0) -- (5,0,0) -- (5,0,5) -- (0,0,5) -- cycle;
    
    \coordinate (goal) at (4,3,-4);
    \node[left] at ($(goal)+(-0.3,0.2,0)$){$\mathscr{G}$};
    \node[below, text=black] at ($(goal)+(0.2,0,0.2)$){$O_{g}$}; 
    \draw[->, thick, green, rotate around x=180, rotate around y=90] (goal) -- ++(0,-2,0) node[right]{$\underline{y}_{g}$}; 
    \draw[->, thick, blue, rotate around x=180, rotate around y=-90] (goal) -- ++(0,0,-2) node[right]{$\underline{z}_{g}$}; 
    \draw[->, thick, red, rotate around x=180, rotate around y=-90] (goal) -- ++(2,0, 0) node[right]{$\underline{x}_{g}$}; 
    
    \coordinate (normal) at (2.5,0,2.5);
    \draw[->,thick, purple] (normal) -- ++(0,1,0) node[right]{$\underline{n}_{s}$};
    
    \coordinate (s1) at (1.5,0,0.7);
    \coordinate (s2) at (4.5,0,1);
    \coordinate (s3) at (3.5,0,4.5);
    \coordinate (s4) at (0.5,0,4);
    
    \coordinate (Is) at (2.5,0,2.5);

    \draw[fill=orange] (s1) circle (3pt) node[left]{$s_{1}$};
    \draw[fill=orange] (s2) circle (3pt) node[right]{$s_{2}$};
    \draw[fill=orange] (s3) circle (3pt) node[right]{$s_{3}$};
    \draw[fill=orange] (s4) circle (3pt) node[left]{$s_{4}$};
    
    \draw[fill=red] (Is) circle (2pt) node[right]{$I_{s}$};
    
    Connect Features to form a rhombus
    \draw[thin, dashed, orange!60] (s1) -- (s2) -- (s3) -- (s4) -- cycle;
    
    \draw[thin, dashed, orange!60] (s1) -- (s3);
    \draw[thin, dashed, orange!60] (s2) -- (s4);

    \draw[->, thick, purple] (camera) -- (goal) node[midway, below]{$\underline{p}_{c}^{g}$};

    \draw[dashed, ->, thick, blue] (camera) -- (s1) node[midway, below]{$\underline{p}_{c}^{s_{1}}$};
    \draw[dashed, ->, thick, blue] (camera) -- (s2) node[midway, above]{$\underline{p}_{c}^{s_{2}}$};
    \draw[dashed, ->, thick, blue] (camera) -- (s3) node[midway, below]{$\underline{p}_{c}^{s_{3}}$};
    \draw[dashed, ->, thick, blue] (camera) -- (s4) node[midway, left]{$\underline{p}_{c}^{s_{4}}$};
    
    \draw[dashed, ->, thick, green!70!black] (goal) -- (s1) node[midway, right]{$\underline{p}_{g}^{s_{1}}$};
    \draw[dashed, ->, thick, green!70!black] (goal) -- (s2) node[midway, right]{$\underline{p}_{g}^{s_{2}}$};
    \draw[dashed, ->, thick, green!70!black] (goal) -- (s3) node[midway, right]{$\underline{p}_{g}^{s_{3}}$};
    \draw[dashed, ->, thick, green!70!black] (goal) -- (s4) node[midway, above]{$\underline{p}_{g}^{s_{4}}$};
    
    \draw[fill=black] (camera) ++(0.5,0,0) -- ++(0,-0.5,0) -- ++(0,0,0.5) -- cycle;
\end{tikzpicture}
    \caption{Camera tracking four planar features on an object while moving from $\mathscr{C}$ to $\mathscr{G}$.}
    \label{fig:cameraFeatureTracking}
\end{figure}
Consider a monocular camera that tracks features on a stationary object while the features are within its FOV using techniques similar to those in \cite{SCC.Bouguet.ea2001, SCC.Lucas.Kanade1981}. Leveraging these tracked features, the camera estimates the relative distances between the features and itself, subsequently utilizing these estimates to reach a user-specified goal location. Let the world frame be a fixed inertial reference frame, denoted by $\mathscr{W} \coloneqq \{\vec{x}_{w}, \vec{y}_{w}, \vec{z}_{w}\}$, with its origin located at $O_{w}$. Let the camera frame, denoted by  $\mathscr{C} \coloneqq \{\vec{x}_{c}, \vec{y}_{c}, \vec{z}_{c}\}$ be fixed to the camera, with its origin $O_{c}$ located at the principal point of the camera. Let the goal frame $\mathscr{G} \coloneqq \{\vec{x}_{g}, \vec{y}_{g}, \vec{z}_{g}\}$, be fixed to the goal, with its origin located at $O_{g}$. The three frames are illustrated in Figure~\ref{fig:cameraFeatureTracking}. 

To facilitate the development of the camera model, the following assumptions are necessary. \begin{assumption}\label{ass:trackableFeatures}
    The stationary object has features that can be detected and tracked, provided it is within the camera's FOV. Specifically, $\forall t \in \mathbb{R}_{\geq 0}$, a set of at least four trackable planar features are in the camera's FOV\cite{SCC.Parikh.Kamalapurkar.ea2015}.
\end{assumption}
\begin{assumption}\label{ass:goalLocation}
While the goal location may lie outside the camera's FOV, the position of the $i$th feature on the object relative to the goal position, denoted by $\underline{p}_{g}^{s_{i}} 
\in \mathbb{R}^{3}$, is known.
\end{assumption}
\begin{assumption}\label{ass:intrinsicMatrix}
    The camera's intrinsic matrix $A \in \mathbb{R}^{3 \times 3}$ is known and invertible \cite{SCC.Ma.Soatto.ea2004}.
\end{assumption}

Given a stationary object $s$ with its $i$th feature denoted by $s_i$ for all $i = 1, \dots, n$, Assumptions~\ref{ass:trackableFeatures} and \ref{ass:goalLocation}, the position of the $i$th feature on the object relative to the goal,  $\underline{p}_{g}^{s_{i}} \in \mathbb{R}^{3}$, is known. Consequently, the position of the goal relative to the camera can be determined as 
\begin{equation}\label{eq:goalPosition}
    \underline{p}_{c}^{g}(t) = \underline{p}_{c}^{s_{i}}(t) - R_{c}^{g}(t)\underline{p}_{g}^{s_{i}},
\end{equation}
where $\underline{p}_{c}^{g}(t) \in \mathbb{R}^{3}$ denotes the unknown position of the goal relative to the camera, $\underline{p}_{c}^{s_{i}}(t) \in \mathbb{R}^{3}$ denotes the unknown position of the $i$th feature on the object, with respect to $\mathscr{C}$, and $R_{c}^{g}(t)  \in \mathbb{R}^{3 \times 3}$ is the rotation matrix describing the orientation of $\mathscr{G}$ with respect to $\mathscr{C}$. The kinematics of the moving monocular camera relative to the goal location are given by
\begin{equation}\label{eq:cameraDynamics}
     \dot{\underline{p}}_{c}^{g}(t) = \underline{v}_{c}^{g}(t) \text{ and } 
    \dot{q}_{c}^{g}(t) = \frac{1}{2} B(q_{c}^{g}(t))\underline{\omega}_{c}^{g}(t),
\end{equation}
where $\underline{v}_{c}^{g}(t) \in \mathbb{R}^{3}$ and $\underline{\omega}_{c}^{g}(t) \in \mathbb{R}^{3}$ represent the linear (unknown) and angular (known) velocities of $\mathscr{G}$ with respect to $\mathscr{C}$, respectively, $B(q(t)) \coloneqq \begin{bmatrix}
    -{q_{v}}^{\top} \\ {q_{0}}\eye_{3 \times 3} + {q_{v}^{\times}}
\end{bmatrix} \in \mathbb{R}^{4 \times 3}$ is an orthogonal matrix which has the pseudoinverse property $B^{\top}(q(t))B(q(t)) = \eye_{3 \times 3}$, where $\eye_{3 \times 3}$ is a $3$ by $3$ identity matrix, and $q_{c}^{g}(t) \in \mathbb{R}^{4}$ represents the quaternion parametization of the rotational matrix $R_{c}^{g}(t)$, describing the orientation of $\mathscr{G}$ with respect to $\mathscr{C}$, with $q \coloneqq \begin{bmatrix}
    q_{0} & q_{v}^{\top}
\end{bmatrix}^{\top} \in \mathcal{S}^{4}$ which has the
standard basis $\{1,i,j,k\}$, where $\mathcal{S}
^{4} \coloneqq \{x \in \mathbb{R}^{4}\vert x^{\top}x = 1\}$, and $q_{0} \in \mathbb{R}$ and $q_{v} \in \mathbb{R}^{3}$ represent the scalar and vector
components of $q$, respectively. The angular velocity of $\mathscr{G}$ with respect to $\mathscr{C}$ given as $\underline{\omega}_{c}^{g}(t)$ is assumed to be known for the rest of the development of the paper.

Equation \eqref{eq:goalPosition} can be equivalently expressed in the form
$\label{eq:unitVector}
\begin{bmatrix}
    \underline{u}_{c}^{s_{i}}(t) & -\underline{u}_{c}^{g}(t)
\end{bmatrix}\begin{bmatrix}
    d_{c}^{s_{i}}(t) \\ d_{c}^{g}(t) 
\end{bmatrix} =  R_{c}^{g}(t)\underline{u}_{g}^{s_{i}}d_{g}^{s_{i}}$, by rearranging the terms, where $d_{c}^{s_{i}}(t) \in \mathbb{R}_{> 0}$ and $\underline{u}_{c}^{s_{i}}(t) \in \mathbb{R}^{3}$ are the magnitude and direction of the position vector $\underline{p}_{c}^{s_{i}}(t)$ of feature $s_{i}$ expressed in $\mathscr{C}$, respectively; $d_{c}^{g}(t) \in \mathbb{R}_{> 0}$ and $\underline{u}_{c}^{g}(t) \in \mathbb{R}^{3}$ are the magnitude and direction of the position vector $\underline{p}_{c}^{g}(t)$ of the goal $\mathscr{G}$ expressed in $\mathscr{C}$, respectively; and $d_{g}^{s_{i}} \in \mathbb{R}_{> 0}$ and $\underline{u}_{g}^{s_{i}} \in \mathbb{R}^{3}$ are the magnitude and direction of the position vector $\underline{p}_{g}^{s_{i}}$ of feature $s_{i}$ expressed in $\mathscr{C}$, respectively. Under Assumptions~\ref{ass:trackableFeatures} - \ref{ass:intrinsicMatrix}, the rotation matrix $R_{c}^{g}(t)$ and unit vector $\underline{u}_{c}^{g}(t)$ can be determined from a general set of features on the object using techniques such as planar homography decomposition or essential decomposition. In addition,  the unit vectors $\underline{u}_{g}^{s_{i}}$ and $\underline{u}_{c}^{s_{i}}(t)$ can be obtained from $\underline{u}_{g}^{s_{i}} \coloneqq  \frac{P_{g}^{s_{i}}}{\|P_{g}^{s_{i}}\|}$ and $\underline{u}_{c}^{s_{i}}(t) \coloneqq  \frac{A^{-1}P_{c}^{s_{i}}(t)}{\|A^{-1}\underline{p}_{c}^{s_{i}}(t)\|}$ where $P_{g}^{s_{i}}$, $P_{c}^{s_{i}}(t)$ are the homogeneous coordinates of feature $s_{i}$ in $\mathscr{G}$ and $\mathscr{C}$, respectively. The only remaining unknowns are the distances $d_{c}^{s_{i}}(t)$, $d_{c}^{g}(t)$ and $d_{g}^{s_{i}}$. These unknowns are estimated in the following using an ICL-based observer.
To simplify the notation, let $H_{s_{i}} (t) \coloneqq \begin{bmatrix}
\underline{u}_{c}^{s_{i}}(t) & -\underline{u}_{c}^{g}(t)
\end{bmatrix} \in \mathbb{R}^{3 \times 2}$. While $d_{c}^{g}(t) > 0$, the term $H_{s_{i}}^{\top}(t)H_{s_{i}}(t) $ is invertible such that (cf. \cite{SCC.Bell.Deptula.ea2020})
\begin{equation}\label{eq:distanceRegressor}
    \begin{bmatrix}
        d_{c}^{s_{i}}(t) \\ d_{c}^{g}(t)
    \end{bmatrix} = Y_{s_{i}}(t)d_{g}^{s_{i}},
\end{equation}
where $ Y_{s_{i}}(t) \coloneqq \left(H_{s_{i}}^{\top}(t)H_{s_{i}}(t)\right)^{-1}H_{s_{i}}^{\top}(t)R_{c}^{g}(t)\underline{u}_{g}^{s_{i}}$ is invertible and measurable under Assumptions~\ref{ass:trackableFeatures}-\ref{ass:intrinsicMatrix}. Furthermore, since the goal and the stationary object are stationary, the time derivatives of the unknown distances are known and given by 
\begin{equation}\label{eq:d1Deriv}
\dot{d}_{c}^{s_{i}}(t) 
 = -{\underline{u}_{c}^{s_{i}}}^{\top}(t)
       \underline{v}_{c}(t),
\end{equation}
\begin{equation}\label{eq:d2Deriv}
    \dot{d}_{c}^{g}(t) 
     = -{\underline{u}_{c}^{g}}^{\top}(t)
        \underline{v}_{c}(t), \text{ and}
\end{equation}
\begin{equation}\label{eq:d3Deriv}
    \dot{d}_{g}^{s_{i}} = 0,
\end{equation}
where $\underline{v}_{c}(t)\in \mathbb{R}^{3}$ represents the velocity of the camera, expressed in $\mathscr{C}$.
Since the goal location is fixed in $\mathscr{W}$, the relationship between $\underline{v}_{c}(t)$ and $\underline{v}_{c}^{g}(t)$ is given by $\underline{v}_{c}(t) = -\underline{v}_{c}^{g}(t)$.



The control objective is to design the camera velocity $\underline{v}_{c}(t)$ to improve feature observability by maximizing orthogonal motion of the camera with respect to the plane containing the features. The objective is achieved by using the ICL-based observer to generate estimates of the distances denoted by $\hat{d}_{c}^{s_{i}}(t) \in \mathbb{R}$, $\hat{d}_{c}^{g}(t) \in \mathbb{R}$ and $\hat{d}_{g}^{s_{i}}(t) \in \mathbb{R}$. Using these estimates and given the known position of the $i$th feature of the stationary object relative to the goal location, the position of the goal relative to the camera expressed in $\mathscr{C}$, $\hat{\underline{p}}_{c}^{g}$, can be estimated using
\begin{equation}\label{eq:goalPositionEstimate}
    \hat{\underline{p}}_{c}^{g}(t) = \hat{\underline{p}}_{c}^{s_{i}}(t) - R_{c}^{g}(t)\underline{p}_{g}^{s_{i}},
\end{equation}
where $\hat{\underline{p}}_{c}^{s_{i}}(t) \in \mathbb{R}^{3}$ denotes the estimate of the position of the $i$th feature on the object with respect to $\mathscr{C}$. Let $I_{s}$ denote the origin of the feature frame and select any three features out of the number of features on the plane that surrounds $I_{s}$ as depicted in Figure~\ref{fig:cameraFeatureTracking}. Let $\underline{n}_{s} \in \mathbb{R}^{3}$ represents the normal vector to the plane containing $s_{1}$, $s_{2}$, $s_{3}$, and $I_{s}$ which can be expressed as $\underline{n}_{s} = \left(\underline{p}_{w}^{s_{1}}-\underline{p}_{w}^{s_{2}}\right) \times \left(\underline{p}_{w}^{s_{3}}-\underline{p}_{w}^{s_{2}}\right)$,
 where $\underline{p}_{w}^{s_{1}}$, $\underline{p}_{w}^{s_{2}}$, and $\underline{p}_{w}^{s_{3}} \in \mathbb{R}^{3}$
 represents the
position of $s_{1}$, $s_{2}$ and $s_{3}$ expressed in $\mathscr{W}$, respectively, and the notation $\times$ represents the cross product. An optimal control problem is then formulated to generate the desired linear velocity commands $\underline{v}_{c}^{g}$ for the camera, online, to minimize the cost functional
\begin{equation} \label{eq:costFunction}
J(\underline{p}_{c}^{g}(\cdot), \underline{v}_{c}^{g}(\cdot)) = \int_{0}^\infty  r_{\scalemath{0.5}{LQR}}(\underline{p}_{c}^{g}(\tau),  \underline{v}_{c}^{g}(\tau)) + r_{\scalemath{0.5}{ORTHO}}\left(  \underline{v}_{c}^{g}(\tau)\right) d\tau,
\end{equation} 
 over the set $\mathcal{U}$ of piece-wise continuous functions and under the dynamic constraint in \eqref{eq:cameraDynamics}. The linear quadratic regulator (LQR) cost denoted by $r_{\scalemath{0.5}{LQR}}: \mathbb{R}^{3} \times \mathbb{R}^{3} \to \mathbb{R}$ is designed to drive the camera to the goal while the orthogonality cost denoted by $r_{\scalemath{0.5}{ORTHO}}: \mathbb{R}^{3} \to \mathbb{R}$ is designed to improve estimates of $\underline{p}_{c}^{g}$ by encouraging orthogonal motion of the camera to the plane that contains the features. The LQR cost $r_{\scalemath{0.5}{LQR}}$ is defined as $\label{eq:LQRCost}
   r_{\scalemath{0.5}{LQR}}(\underline{p}_{c}^{g}(t), \underline{v}_{c}^{g}(t)) \coloneqq \underline{p}_{c}^{g\top}(t)Q_{c}\underline{p}_{c}^{g}(t) +  \underline{v}_{c}^{g\top}(t)R_{c}\underline{v}_{c}^{g}(t)$,  
 where $Q_{c} \in \mathbb{R}^{3 \times 3}$ and $R_{c} \in \mathbb{R}^{3 \times 3}$ are constant positive definite symmetric matrices. The orthogonality cost $r_{\scalemath{0.5}{ORTHO}}$ is designed as $\label{eq:OrthoCost} r_{\scalemath{0.5}{ORTHO}}(\underline{v}_{c}^{g}(t)) \coloneqq \gamma_{c}\left(\langle \underline{v}_{c}^{g}(t), \underline{n}_{s} \rangle \right)^{2}$,  where $\gamma_{c} \in \mathbb{R}_{> 0}$ is a user-defined constant designed to maximize orthogonality of the motion of the camera relative to the feature plane. The goal to move the camera orthogonally relative to the feature plane is captured in \eqref{eq:costFunction} via minimization of the dot product $\langle \underline{v}_{c}^{g}(t), \underline{n}_{s}\rangle$.


\section{ICL-Based Observer Design}\label{section:ICLObserver}
An ICL update law is implemented to estimate the unknown distances $d_{c}^{s_{i}}(t)$, $d_{c}^{g}(t)$,  and $d_{g}^{s_{i}}$ by integrating \eqref{eq:d1Deriv}, \eqref{eq:d2Deriv} and \eqref{eq:d3Deriv}, respectively,  over a time delay $T \in \mathbb{R}_{> 0}$ to obtain
{\small\begin{equation}
    \begin{bmatrix}
        d_{c}^{s_{i}}(t) \\ d_{c}^{g}(t) 
    \end{bmatrix} - \begin{bmatrix}
        d_{c}^{s_{i}}(t-T) \\ d_{c}^{g}(t-T) 
    \end{bmatrix} =  - \int_{t-T}^{t} \begin{bmatrix}
        {\underline{u}_{c}^{s_{i}}}^{\top}(\tau) \\ {\underline{u}_{c}^{g}}^{\top}(\tau)
        \end{bmatrix}\underline{v}_{c}(\tau) d\tau, 
\end{equation}}for $t > T$. Substituting the relationship in equation \eqref{eq:distanceRegressor} at current time $t$ and previous time $t-T$ yields
\begin{equation}\label{eq:augDynamics}
    \mathcal{Y}_{s_{i}}(t)d_{g}^{s_{i}} = \mathcal{U}_{s_{i}}(t)
\end{equation}
where $\mathcal{Y}_{s_{i}} \coloneqq \begin{cases}
    0_{2 \times 1}, & t \leq T, \\
    Y_{s_{i}}(t)-Y_{s_{i}}(t-T), & t > T, 
\end{cases}$
and $\mathcal{U}_{s_{i}}(t) \coloneqq \begin{cases}
    0_{2 \times 1}, & t \leq T, \\
    - \int_{t-T}^{t} \begin{bmatrix}
        {\underline{u}_{c}^{s_{i}}}^{\top}(\tau) \\ {\underline{u}_{c}^{g}}^{\top}(\tau)
        \end{bmatrix}\underline{v}_{c}(\tau) d\tau, & t > T.
\end{cases}$ Multiplying both sides of \eqref{eq:augDynamics} by the term $\mathcal{Y}_{s_{i}}^{\top}(t)$ yields
 \begin{equation}\label{eq:augDynamics2}
    \mathcal{Y}_{s_{i}}^{\top}(t)\mathcal{Y}_{s_{i}}(t)d_{g}^{s_{i}} = \mathcal{Y}_{s_{i}}^{\top}(t)\mathcal{U}_{s_{i}}(t)
\end{equation}
In general, $\mathcal{Y}_{s_{i}}(t)$ will not have full column rank (e.g. when the camera
is stationary) implying $\mathcal{Y}_{s_{i}}^{\top}(t)\mathcal{Y}_{s_{i}}(t)$ is positive semidefinite but not positive definite. However, the equality in \eqref{eq:augDynamics2} may be evaluated at several (possibly time-varying)  time instances $t_{1}, \hdots, t_{N}$ and summed together to yield
\begin{equation}\label{eq:CLdyanmics}
    \Sigma_{\mathcal{Y}_{s_{i}}}(t)d_{g}^{s_{i}} = \Sigma_{\mathcal{U}_{s_{i}}}(t)
\end{equation}
where $\Sigma_{\mathcal{Y}_{s_{i}}}(t) \coloneqq \sum_{j = 1}^{N} \mathcal{Y}_{s_{i}}^{\top}(t_{j}(t))\mathcal{Y}_{s_{i}}(t_{j}(t))$, $\Sigma_{\mathcal{U}_{s_{i}}}(t) \coloneqq \sum_{j = 1}^{N} \mathcal{Y}_{s_{i}}^{\top}(t_{j}(t))\mathcal{U}_{s_{i}}(t_{j}(t))$, and $N \in \mathbb{Z}_{\geq 1}$. The following assumption is an observability-like condition that must be satisfied to guarantee convergence of distance estimates in finite time.
\begin{assumption}\label{ass:sufficientExcitation}
    The camera has sufficiently rich motion so that there exist constants $\tau \in \mathbb{R}_{>T}$ and $\lambda_{\tau} \in \mathbb{R}_{> 0}$ such that for all $t \geq \tau$, $\lambda_{\min} \{\Sigma_{\mathcal{Y}_{s_{i}}}(t)
    \} > \lambda_{\tau}$, where $\lambda_{\min}\{\cdot\}$ denotes the minimum eigenvalue of $\{\cdot\}$.
\end{assumption}
\begin{remark}
    Assumption~\ref{ass:sufficientExcitation} can be verified online and is easy to satisfy provided the trajectories contain sufficient information to make $\mathcal{Y}_{s_{i}}$ sufficiently exciting on a finite interval\cite{SCC.Parikh.Kamalapurkar.ea2015, SCC.Bell.Parikh.ea2016, SCC.Bell.Deptula.ea2020}.
\end{remark}

The time $\tau$ is unknown; however, it can be determined online by checking the minimum eigenvalue of $\Sigma_{\mathcal{Y}_{s_{i}}}(t)$. After $t =\tau$, $\lambda_{\min}\{\Sigma_{\mathcal{Y}_{s_{i}}}(t)
\} > \lambda_{\tau}$ implies that the constant unknown distance $d_{g}^{s_{i}}$ can be determined from \eqref{eq:CLdyanmics} and obtained as
$\label{eq:unknownDistsk}
d_{g}^{s_{i}} = \begin{cases}
    0,& t < \tau, \\
    \Sigma_{\mathcal{Y}_{s_{i}}}^{-1}(t)\Sigma_{\mathcal{U}_{s_{i}}}(t), & t \geq \tau.
\end{cases}$ Substituting this expression into \eqref{eq:distanceRegressor} yields $\label{eq:unknowndc}
\begin{bmatrix}
    d_{c}^{s_{i}}(t) \\ d_{c}^{g}(t)
\end{bmatrix} =  \begin{cases}
    0,& t < \tau, \\
    Y_{s_{i}}(t)\Sigma_{\mathcal{Y}_{s_{i}}}^{-1}(t)\Sigma_{\mathcal{U}_{s_{i}}}(t), & t \geq \tau.
\end{cases}$ Based on subsequent stability analysis, ICL update laws to generate the estimates $\hat{d}_{c}^{s_{i}}(t)$, $\hat{d}_{c}^{g}(t)$, and $\hat{d}_{g}^{s_{i}}$ are designed as
\begin{equation}\label{eq:d1H}
    \dot{\hat{d}}_{c}^{s_{i}}(t) {\coloneqq} \begin{cases}
        \eta_{s_{i}, 1}(t), & t < \tau, \\
        \eta_{s_{i}, 1}(t) + \kappa_{1}\left(\nu_{s_{i}, 1}(t)- \hat{d}_{c}^{s_{i}}(t)\right), & t \geq \tau,
    \end{cases}
\end{equation}
\begin{equation}\label{eq:d2H}
    \dot{\hat{d}}_{c}^{g}(t) {\coloneqq} 
    \begin{cases}
        \eta_{s_{i}, 2}(t), & t < \tau, \\
        \eta_{s_{i}, 2}(t) + \kappa_{2}\left(\nu_{s_{i}, 2}(t)- \hat{d}_{c}^{g}(t)\right), & t \geq \tau,
    \end{cases}
\end{equation}
and
\begin{equation}\label{eq:d3H}
    \dot{\hat{d}}_{g}^{s_{i}}(t) \coloneqq \begin{cases}
        0, & t < \tau, \\
        \kappa_{3}\left(\Sigma_{\mathcal{Y}_{s_{i}}}^{-1}(t)\Sigma_{\mathcal{U}_{s_{i}}}(t) -\hat{d}_{g}^{s_{i}}(t)\right), & t \geq \tau,
    \end{cases}
\end{equation}
respectively, where $\eta_{s_{i}}(t) \coloneqq -\begin{bmatrix}
{\underline{u}_{c}^{s_{i}}}^{\top}(t) \\ {\underline{u}_{c}^{g}}^{\top}(t)
\end{bmatrix}\underline{v}_{c}(t)$, $\nu_{s_{i}}(t) \coloneqq Y_{s_{i}}(t)\Sigma_{\mathcal{Y}_{s_{i}}}^{-1}(t)\Sigma_{\mathcal{U}_{s_{i}}}(t)$, and $\kappa_{1} \in \mathbb{R}_{> 0}$, $\kappa_{2}\in \mathbb{R}_{> 0}$, and $\kappa_{3}\in \mathbb{R}_{> 0}$ are user-selected gains. Let $\tilde{d}_{c}^{s_{i}}(t) \in \mathbb{R}, \tilde{d}_{c}^{g}(t) \in \mathbb{R}$ and $\tilde{d}_{g}^{s_{i}}(t) \in \mathbb{R}$ represent the distance estimation errors defined as $\tilde{d}_{c}^{s_{i}}(t) \coloneqq d_{c}^{s_{i}}(t) - \hat{d}_{c}^{s_{i}}(t)$, $\tilde{d}_{c}^{g}(t) \coloneqq d_{c}^{g}(t) - \hat{d}_{c}^{g}(t)$ and $\tilde{d}_{g}^{s_{i}}(t) \coloneqq d_{g}^{s_{i}} - \hat{d}_{g}^{s_{i}}(t)$, respectively. Taking their derivatives and substituting the dynamics in \eqref{eq:d1Deriv}, \eqref{eq:d2Deriv}, and \eqref{eq:d3Deriv} and update laws in \eqref{eq:d1H}, \eqref{eq:d2H}, and \eqref{eq:d3H} yields 
\begin{equation}\label{eq:d1Error}
    \dot{\tilde{d}}_{c}^{s_{i}}(t) {\coloneqq} \begin{cases}
        0, & t < \tau, \\
        -\kappa_{1}\tilde{d}_{c}^{s_{i}}(t), & t \geq \tau,
    \end{cases}
\end{equation}
\begin{equation}\label{eq:d2Error}
    \dot{\tilde{d}}_{c}^{g}(t) {\coloneqq} \begin{cases}
        0, & t < \tau, \\
        -\kappa_{2}\tilde{d}_{c}^{g}(t), & t \geq \tau,
    \end{cases}
\end{equation}
and
\begin{equation}\label{eq:d3Error}
    \dot{\tilde{d}}_{g}^{s_{i}}(t) \coloneqq \begin{cases}
        0, & t < \tau, \\
        -\kappa_{3}\tilde{d}_{g}^{s_{i}}(t), & t \geq \tau,
    \end{cases}
\end{equation}
The subsequent analysis in Section~\ref{section:stabilityAnalysis} shows that the error $\tilde{d}_{c}^{s_{i}}$ remains bounded for $t < \tau$ and decays exponentially for $t \geq \tau$, once sufficient data has been gathered. 

\section{Design of feature observability maximizing Velocity}\label{section:controlDesign}

This section presents an analytical solution to the optimal control problem in \eqref{eq:costFunction} using estimates of the position of the goal relative to the camera $\hat{\underline{p}}_{c}^{g}(t)$ obtained from the results of the observer in Section~\ref{section:ICLObserver}. The Hamilton-Jacobi-Bellman (HJB) equation for the optimal control problem in \eqref{eq:costFunction} can be expressed in the form,
{\small\begin{multline}\label{eq:HJB} 
0 = \min_{ \underline{v}_{c}^{g}} \bigg\{J^{*\prime}(\underline{p}_{c}^{g}) \underline{v}_{c}^{g}(\underline{p}_{c}^{g}) + \underline{p}_{c}^{g\top}Q_{c}\underline{p}_{c}^{g} 
+  {{\underline{v}_{c}^{g*\top}}}(\underline{p}_{c}^{g})R_{c} \underline{v}_{c}^{g*}(\underline{p}_{c}^{g}) \\ + \gamma_{c}\left(\langle\underline{v}_{c}^{g*}(\underline{p}_{c}^{g}), \underline{n}_{s}\rangle\right)^{2} \bigg\},
\end{multline}}where $J^{*}: \mathbb{R}^{3} \to \mathbb{R}$ is the optimal cost-to-go. Since the position dynamics in \eqref{eq:cameraDynamics} are linear and the cost in \eqref{eq:costFunction} is quadratic, the optimal cost-to-go is given by $\label{eq:valuefunction} J^{*}(\underline{p}_{c}^{g}) \coloneqq \underline{p}_{c}^{g\top}S_{c}\underline{p}_{c}^{g}$, where $S_{c} \in \mathbb{R}^{3 \times 3}$ is a constant positive
definite symmetric matrix, and the notation  $\left(\cdot\right)^{\prime}$ is used to denote $\frac{\partial}{\partial \left(\cdot\right)}$. The optimal control policy, denoted by $ \underline{v}_{c}^{g*}: \mathbb{R}^{3} \to \mathbb{R}^{3}$, is given as
\begin{equation}\label{eq:optimalcontrol}
     \underline{v}_{c}^{g*}(\underline{p}_{c}^{g}(t)) = -\overline{R}_{c}^{-1}S_{c}\underline{p}_{c}^{g}(t),
\end{equation}
where $\overline{R}_{c} \in \mathbb{R}^{3 \times 3}$ is a positive definite matrix defined as $\overline{R}_{c} \coloneqq R_{c}+\gamma_{c}N_{s}$ and $N_{s} \in \mathbb{R}^{3 \times 3}$ is a positive semi-definite symmetric defined as $N_{s} \coloneqq \underline{n}_{s}\underline{n}_{s}^{\top}$.
Since the matrix $\overline{R}_{c}$ is the sum of a symmetric positive definite matrix and a symmetric positive semi-definite matrix, it is also symmetric and positive definite. Substituting the \eqref{eq:optimalcontrol} back into the HJB \eqref{eq:HJB} equation and simplifying yields the following necessary and sufficient condition for optimality
\begin{equation}\label{eq:are}
    -S_{c}\overline{R}_{c}^{-1}S_{c} + Q_{c} = 0,
\end{equation}
where the objective is to find the matrix $S_{c}$. 
Given symmetric positive semi-definite matrices $\overline{R}_{c}$, $Q_{c}$ and $S_{c}$, the solution to the quadratic equation in \eqref{eq:are} is unique and is given as 
\begin{equation}\label{eq:ScSolution}
S_{c} = \overline{R}_{c}^{1/2}(\overline{R}_{c}^{-1/2}Q_{c}\overline{R}_{c}^{-1/2})^\frac{1}{2}\overline{R}_{c}^{1/2}.
\end{equation}
Since $\underline{p}_{c}^{g}(t)$ is unknown, the linear velocity of the camera is subsequently designed using the estimate $\hat{\underline{p}}_{c}^{g}(t)$ as 
\begin{equation}\label{eq:uControl}
    \underline{v}_{c}(t) \coloneqq -\hat{\underline{v}}_{c}^{g}(t) = K_{s}\hat{\underline{p}}_{c}^{g}(t),
\end{equation}
{where $K_{s} \in \mathbb{R}^{3 \times 3}$ is the feedback gain defined as $K_{s} \coloneqq \overline{R}_{c}^{-1}S_{c}$. The velocity $\underline{v}_{c}(t)$, when represented in $\mathscr{W}$, is denoted by $\underline{v}_{w}^{c}(t) \in \mathbb{R}^{3}$ and given by $\underline{v}_{w}^{c}(t) =  K_{s}R_{w}^{c}(t)\hat{\underline{p}}_{c}^{g}(t)$ where $R_{w}^{c}(t)$ is the orientation of $\mathscr{C}$ with respect to $\mathscr{W}$.  


\section{Stability Analysis}\label{section:stabilityAnalysis}

This section presents the main theoretical results of this paper. 
First, the convergence properties of the proposed observers in Section~\ref{section:ICLObserver} are presented, and finally, the convergence of the position error trajectory $\underline{p}_{c}^{g}(t)$ to a given neighborhood of the origin is presented.

\subsection{Analysis of Camera ICL Observer Error system}\label{section:analysisOfObsever}
Let $\tilde{\vartheta}(t) \in \mathbb{R}^{9}$ denote
a concatenated state vector containing the distance estimation errors, defined as $ \tilde{\vartheta}(t) \coloneqq \begin{bmatrix}
    \tilde{d}_{c}^{s_{i}}(t) & \tilde{d}_{c}^{g}(t) &  \tilde{d}_{g}^{s_{i}}(t)
\end{bmatrix}^{\top}$ and let $L:\mathbb{R}^{9}\to \mathbb{R}$ be a candidate Lyapunov function defined as
\begin{equation}\label{eq:distLyapFunc}
    L(\tilde{\vartheta}(t)) = \frac{1}{2}\tilde{\vartheta}^{\top}(t)\tilde{\vartheta}(t).
\end{equation}
The following theorem establishes the exponential stability of the observer error system obtained in \eqref{eq:d1Error}, \eqref{eq:d2Error}, and \eqref{eq:d3Error}.
\begin{theorem}\label{thm:expObserver}
    Provided Assumptions~\ref{ass:trackableFeatures}-\ref{ass:sufficientExcitation} hold, the update laws defined in \eqref{eq:d1H}, \eqref{eq:d2H}, and \eqref{eq:d3H} ensure that the origin of the observer error system is globally exponentially stable and the trajectories of the estimation errors $\tilde{\vartheta}(\cdot)$ converge exponentially to the origin.
\end{theorem}
\begin{proof}
    Taking the orbital derivative of the candidate Lyapunov function in \eqref{eq:distLyapFunc}, along the solutions of  \eqref{eq:d1Error}, \eqref{eq:d2Error}, and \eqref{eq:d3Error}, simplifying, and upper bounding, yields the inequality
    {\small\begin{equation}\label{eq:distLyapDeriv}
        \dot{L}(\tilde{\vartheta}(t)) \leq \begin{cases}
            0, & t < \tau,
            \\
            -2\kappa L(\tilde{\vartheta}(t)), & t \geq \tau,
        \end{cases}
    \end{equation}}where $\kappa = \min\{\kappa_{1}, \kappa_{2}, \kappa_{3}\}$. At $t < \tau$, it can be observed from \eqref{eq:distLyapFunc} and \eqref{eq:distLyapDeriv} that the distance estimation errors in $\tilde{\vartheta}(t)$ are non-increasing, specifically $\tilde{\vartheta}(t) \leq \vartheta(0), \forall t < \tau$. Invoking \cite[Theorem~4.10]{SCC.Khalil2002}, it can be concluded that the observer error system is exponentially stable and by the Comparison Lemma \cite[Lemma~3.4]{SCC.Khalil2002}, the bound $\|\tilde{\vartheta}(t)\| \leq \|\vartheta(\tau)\|e^{-\kappa(t-\tau)}$ holds for all $ t \geq \tau$.
\end{proof}

\subsection{Analysis of position error system}

To facilitate the following analysis, let $\Gamma_{s} \coloneqq S_{c} R_{c}^{-1}S_{c}$ and note that $\lambda_{\min}(\Gamma_{s})\|\underline{p}_{c}^{g}(t)\|^{2} \leq \underline{p}_{c}^{g\top}(t)\Gamma_{s}\underline{p}_{c}^{g}(t) \leq  \lambda_{\max}(\Gamma_{s})\|\underline{p}_{c}^{g}(t)\|^{2}$. Using the optimal cost-to-go function $J^{*}$ as the candidate Lyapunov function}, the following theorem establishes the input-to-state stability of the position error system.
\begin{theorem}\label{thm:ISSTrajectory}
    Provided Assumption~\ref{ass:trackableFeatures}-\ref{ass:sufficientExcitation} hold and the estimated distances $\hat{d}_{c}^{s_{i}}(\cdot)$, $\hat{d}_{c}^{g}(\cdot)$,  and $\hat{d}_{g}^{s_{i}}(\cdot)$ are updated using the update laws defined in \eqref{eq:d1H}, \eqref{eq:d2H}, and \eqref{eq:d3H} respectively such that the conditions of Theorem~\ref{thm:expObserver} are satisfied, then the system in \eqref{eq:cameraDynamics} is input-to-state stable with state $\underline{p}_{c}^{g}(\cdot)$ and input $\sqrt{\|\tilde{d}_{c}^{g}(\cdot)\|}$.
\end{theorem}
\begin{proof}
    The orbital derivative of the optimal cost-to-go function is bounded as $\dot{J}(\underline{p}_{c}^{g}(t)) \leq -2\lambda_{\min}(\Gamma_{s})\|\underline{p}_{c}^{g}(t)\|^{2} + 2\lambda_{\max}(\Gamma_{s})\|\underline{p}_{c}^{g}(t)\|\|\tilde{\underline{p}}_{c}^{g}(t)\|$. Applying completion of squares and using the fact that $\sup_{t\in \mathbb{R}_{\geq 0}}\|\underline{u}_{c}^{g}(t)\| 
 \leq 1$ since $\underline{u}_{c}^{g}(t)$ is a unit vector, the
orbital derivative is bounded for all $t \geq 0$ as
$\dot{J}(\underline{p}_{c}^{g}(t)) \leq -\lambda_{\min}(\Gamma_{s})\|\underline{p}_{c}^{g}(t)\|^{2}, \forall \|\underline{p}_{c}^{g}(t)\| \geq \varrho\scalemath{0.9}{\left(\sqrt{\|\tilde{d}_{c}^{g}(\cdot)\|}\right)}
\label{eq.VDerivBound}$, 
  where $\varrho\scalemath{0.95}{\left(\sqrt{\|\tilde{d}_{c}^{g}(\cdot)\|}\right)} \coloneqq \sqrt{\frac{2\lambda_{\max}(\Gamma_{s})}{\lambda_{\min}(\Gamma_{s})}}\|\tilde{d}_{c}^{g}(\cdot)\|$. Therefore, the conditions of \cite[Theorem~4.19]{SCC.Khalil2002} are satisfied and can be concluded that the system in \eqref{eq:cameraDynamics} is input-to-state stable with state $\underline{p}_{c}^{g}(\cdot)$ and input $\sqrt{\|\tilde{d}_{c}^{g}(\cdot)\|}$. Since the distance error $\tilde{d}_{c}^{g}(\cdot)$ converges exponentially to the origin according to Theorem~\ref{thm:expObserver}, the results of \cite[Exercise~4.58]{SCC.Khalil2002} can be used to show that as $t \to \infty$ and the input $\sqrt{\|\tilde{d}_{c}^{g}(\cdot)\|}$ converges to zero, so does the state $\underline{p}_{c}^{g}(\cdot)$.
\end{proof}


\section{Simulation Study}
    \label{section:simulation}
     To demonstrate the performance of the developed observers and to test the effects of the orthogonality cost $r_{\scalemath{0.5}{ORTHO}}$ on feature observability, consider a monocular camera with dynamics as defined in \eqref{eq:cameraDynamics}, which is tracking four co-planar features on a stationary object as described in Figure~\ref{fig:cameraFeatureTracking}. The control objective is to move the camera from the initial position to the goal location using the control policy in \eqref{eq:uControl}, which uses estimates of the position of the camera obtained from the ICL-based observers developed Section~\ref{section:ICLObserver}. The simulation parameters are omitted for brevity of the paper and are available in the \textit{arXiv} version of this paper.

\begin{figure}
    \centering
     \begin{tikzpicture}[every pin/.style={fill=white}]
    \begin{axis}[
        xlabel={$t$ [s]},
        ylabel={$\tilde{d}_{c}^{s_{i}}(t)$},
        legend pos=outer north east,
        legend style={nodes={scale=0.6, transform shape}},
        enlarge y limits=0,
        enlarge x limits=0,
        height=0.35\columnwidth,
        width=0.8\columnwidth,
        label style={font=\scriptsize},
        tick label style={font=\scriptsize},
        name=mainplot
    ]
    \pgfplotsinvokeforeach{1,...,4}{
        \addplot+ [thick, mark=none, color=color#1] table [x index=0, y index=#1] {data/dicTildeData.dat};
        \addlegendentry{$\tilde{d}_{c}^{s_{#1}}(t)$}
    }
    \end{axis}
\end{tikzpicture}
    \caption{Trajectory of the distance error of the features of the object relative to the camera.}
    \label{fig:dicTilde}
\end{figure}
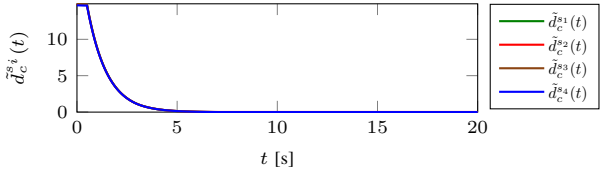
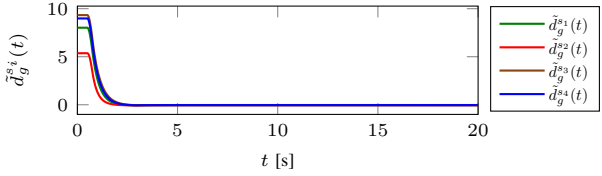
\begin{figure}
    \centering
     \begin{tikzpicture}
    \begin{axis}[
        xlabel={$t$ [s]},
        ylabel={$\tilde{d}_{g}^{s_{i}}(t)$},
        legend pos= outer north east,
        legend style={nodes={scale=0.6, transform shape}},
        enlarge y limits=0.1,
        enlarge x limits=0,
        height = 0.35\columnwidth,
        width = 0.8\columnwidth,
        label style={font=\scriptsize},
        tick label style={font=\scriptsize}
    ]
    \pgfplotsinvokeforeach{1,...,4}{
        \addplot+ [thick,color=color#1, mark=none] table [x index=0, y index=#1] {data/didTildeData.dat};
        \addlegendentry{$\tilde{d}_{g}^{s_{#1}}(t)$}
    }
    \end{axis}
\end{tikzpicture}
    \caption{Trajectory of the distance error of the features of the object relative to the goal.}
    \label{fig:digTilde}
\end{figure}
\begin{figure}
    \centering
     \begin{tikzpicture}
    \begin{axis}[
        xlabel={$t$ [s]},
        ylabel={$\tilde{d}_{c}^{g}(t)$},
        legend pos= outer north east,
        legend style={nodes={scale=0.75, transform shape}},
        enlarge y limits=0.1,
        enlarge x limits=0,
        height = 0.35\columnwidth,
        width = 0.8\columnwidth,
        label style={font=\scriptsize},
        tick label style={font=\scriptsize}
    ]
    \pgfplotsinvokeforeach{1}{
        \addplot+ [thick, mark=none] table [x index=0, y index=#1] {data/ddcTildeData.dat};
    }
    
    \legend{$\tilde{d}_{c}^{g}(t)$}
    \end{axis}
\end{tikzpicture}
    \caption{Trajectory of the distance error of the goal relative to the camera.}
    \label{fig:dgcTilde}
\end{figure}
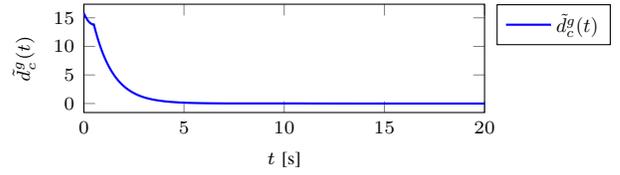
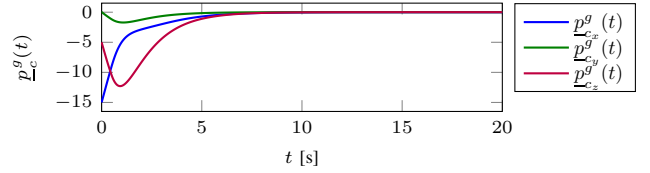
\begin{figure}
    \centering
     \begin{tikzpicture}
    \begin{axis}[
        xlabel={$t$ [s]},
        ylabel={$\underline{p}_{c}^{g}(t)$},
        legend pos= outer north east,
        legend style={nodes={scale=0.75, transform shape}},
        enlarge y limits=0.1,
        enlarge x limits=0,
        height = 0.35\columnwidth,
        width = 0.8\columnwidth,
        label style={font=\scriptsize},
        tick label style={font=\scriptsize}
    ]
    \pgfplotsinvokeforeach{1}{
        \addplot+ [thick, mark=none, color=blue] table [x index=0, y index=#1] {data/Pdcw1HistData.dat};
    }
    \pgfplotsinvokeforeach{1}{
        \addplot+ [thick, mark=none, color=green!50!black] table [x index=0, y index=#1] {data/Pdcw2HistData.dat};
    }

    \pgfplotsinvokeforeach{1}{
        \addplot+ [thick, mark=none, color=purple] table [x index=0, y index=#1] {data/Pdcw3HistData.dat};
    }
    \legend{$\underline{p}_{c_{x}}^{g}(t)$, $\underline{p}_{c_{y}}^{g}(t)$,  $\underline{p}_{c_{z}}^{g}(t)$} 
    \end{axis}
\end{tikzpicture}
    \caption{The trajectories of the actual goal position relative to the camera expressed in $\mathscr{W}$.} 
    \label{fig:pgc}
\end{figure}

\subsection{Results}
From Figures~\ref{fig:dicTilde}, \ref{fig:digTilde}, and \ref{fig:dgcTilde}, it can be observed that the trajectories of the distance errors converge exponentially to the origin which is consistent with the results of Theorem~\ref{thm:expObserver}. 
 Similarly, it can be observed in Figure~\ref{fig:pgcHat} and in Figure~\ref{fig:pgc} that the trajectories of the actual and estimated position of the goal relative to the camera $\underline{p}_{c}^{g}(\cdot)$ and $\hat{\underline{p}}_{c}^{g}(\cdot)$, respectively, decreases and eventually converges to the origin which is consistent with the results of Theorem~\ref{thm:ISSTrajectory}. 

\begin{figure}
    \centering
     \begin{tikzpicture}
    \begin{axis}[
        xlabel={$t$ [s]},
        ylabel={$\hat{\underline{p}}_{c}^{g}(t)$},
        legend pos= outer north east,
        legend style={nodes={scale=0.75, transform shape}},
        enlarge y limits=0.1,
        enlarge x limits=0,
        height = 0.35\columnwidth,
        width = 0.8\columnwidth,
        label style={font=\scriptsize},
        tick label style={font=\scriptsize}
    ]
    \pgfplotsinvokeforeach{1}{
        \addplot+ [thick, solid, mark=none, color=red] table [x index=0, y index=#1] {data/Pdcw1HatHistData.dat};
    }
    \pgfplotsinvokeforeach{1}{
        \addplot+ [thick, solid, mark=none, color=orange] table [x index=0, y index=#1] {data/Pdcw2HatHistData.dat};
    }

    \pgfplotsinvokeforeach{1}{
        \addplot+ [thick, solid, mark=none, color=purple!50!black] table [x index=0, y index=#1] {data/Pdcw3HatHistData.dat};
    }
    
    \legend{$\hat{\underline{p}}_{c_{x}}^{g}(t)$, $\hat{\underline{p}}_{c_{y}}^{g}(t)$, $\hat{\underline{p}}_{c_{z}}^{g}(t)$}
    \end{axis}
\end{tikzpicture}
    \caption{The trajectories of the estimated goal position relative to the camera expressed in $\mathscr{W}$.} 
    \label{fig:pgcHat}
\end{figure}
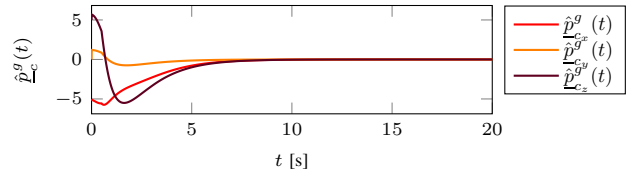

\begin{table}[h!]
    \centering
    \caption{Effect of Varied Orthogonality Penalty Gains ($\gamma_{c}$) on Regressor Conditioning.}
    \label{tab:featureObsTable}
    \begin{tabular}{p{1.75cm}|p{1cm}p{1cm}p{1cm}p{1cm}p{1cm}p{1cm}p{1cm}} 
  \hline 
 $\gamma_{c}$  & 0 & 5 &  10 & 15  & 25 & 50  \\
 \hline
 Avg. Cond. no & 16.289 & 9.906 & 6.199 & 5.239 &  3.279 & 2.718\\
 \hline
  \hline
    \end{tabular}%
\end{table}

Table~\ref{tab:featureObsTable} illustrates how the conditioning of the regressor $\Sigma_{\mathcal{Y}{s{i}}}$, as defined in \eqref{eq:CLdyanmics}, varies with increasing $\gamma_{c}$ values. A lower condition number implies better numerical stability while estimating the Euclidean distance to the features, maximizing feature observability. Conversely, a higher condition indicates heightened sensitivity to measurement errors, indicating that the regressor is poorly conditioned, resulting in less accurate and reliable estimates. The result of  Table~\ref{tab:featureObsTable} demonstrates the impact of the added cost $r_{\scalemath{0.5}{ORTHO}}$ on obtaining better scale estimates; however, increasing the camera gain $\gamma_c$ beyond a certain threshold can negatively affect the performance of the controller in achieving its objective of reaching the goal position. A careful choice of the parameter $\gamma_{c}$ allows for the right tradeoff between maximizing feature observability and achieving the goal, an important consideration in real-time systems, where the camera's objective to reach the goal must be balanced with the need to observe the features of landmarks within the operating environment.

\section{Conclusion}\label{section:conclusion}
This paper develops a technique to plan trajectories for a monocular camera to maximize the observability of the features of a stationary object by formulating an optimal control problem whose objective is to reach a goal location while using estimates generated by ICL-based observers. The developed method does not require the
positive depth constraint, which requires that the distance from the
focal point of the camera to the target along the axis perpendicular to the
image plane must remain positive, or the PE condition, which is difficult to satisfy in practice. As evidenced by the results described in Table~\ref{tab:featureObsTable}, noticeable improvements are obtained due to the added orthogonality cost designed to maximize observability. Future work will involve extending these results to nonplanar features on an object, as well as multiple objects that are non-stationary. 

\small
 
\bibliographystyle{IEEETrans.bst}
\bibliography{scc,sccmaster,scctemp}

\end{document}